\documentclass{article}

\PassOptionsToPackage{numbers, sort}{natbib}

\usepackage[preprint]{neurips_2022}

\usepackage[utf8]{inputenc} %
\usepackage[T1]{fontenc}    %
\usepackage{hyperref}       %
\usepackage{url}            %
\usepackage{booktabs}       %
\usepackage{amsfonts}       %
\usepackage{nicefrac}       %
\usepackage{microtype}      %
\usepackage{xcolor}         %

\usepackage{xspace}
\usepackage{graphicx}

\usepackage{amsthm}
\usepackage{amsmath}
\usepackage{thm-restate}

\usepackage{algorithm}
\usepackage[noend]{algorithmic}

\usepackage[textsize=footnotesize,backgroundcolor=green!10]{todonotes}

\graphicspath{{./figures/}}

\usepackage[sort&compress,nameinlink,capitalize]{cleveref}

\title{The Influence of Dimensions on the Complexity of Computing Decision Trees}

\author{%
  Stephen~G.~Kobourov \\
  \texttt{kobourov@cs.arizona.edu} \\
   \And
   Maarten~L\"offler \\
   \texttt{m.loffler@uu.nl} \\
   \AND
   Fabrizio~Montecchiani \\
   \texttt{fabrizio.montecchiani@unipg.it} \\
   \And
   Marcin~Pilipczuk \\
   \texttt{malcin@mimuw.edu.pl} \\
   \And
   Ignaz~Rutter \\
   \texttt{ignaz.rutter@uni-passau.de} \\
   \And
   Raimund~Seidel \\
   \texttt{rseidel@cs.uni-saarland.de} \\
   \And
   Manuel~Sorge\\
   \texttt{manuel.sorge@ac.tuwien.ac.at} \\
   \And
   Jules~Wulms \\
   \texttt{jwulms@ac.tuwien.ac.at}
}

\theoremstyle{plain}
\newtheorem{theorem}{Theorem}[section]
\newtheorem{lemma}[theorem]{Lemma}
\newtheorem{corollary}[theorem]{Corollary}

\newcommand{\dtslong}{\textsc{Minimum Decision Tree Size}\xspace}
\newcommand{\dts}{\textsc{DTS}\xspace}
\newcommand{\kdts}{$k$-\textsc{DTS}\xspace}

\newcommand{\dimn}{\textsf{dim}}
\newcommand{\thr}{\textsf{thr}}

\newcommand{\lpos}{\ensuremath{\text{blue}}}
\newcommand{\lneg}{\ensuremath{\text{red}}}

\newcommand{\ethlong}{Exponential Time Hypothesis}

\newcommand{\boxx}{\ensuremath{\mathrm{box}}}

\newcommand{\col}{\ensuremath{\textsf{col}}}
\newcommand{\subisolong}{\textsc{Partitioned Subgraph Isomorphism}\xspace}
\newcommand{\subiso}{\textsc{PSI}\xspace}

\newcommand{\lowerbound}{\ensuremath{f(d) \cdot n^{o(d / \log d)}}}

\newcommand{\poly}{\ensuremath{\operatorname{poly}}}

\begin{document}

\maketitle

\begin{abstract}
  A decision tree recursively splits a feature space $\mathbb{R}^{d}$ and then assigns class labels based on the resulting partition.
  Decision trees have been part of the basic machine-learning toolkit for decades.
  A large body of work treats heuristic algorithms to compute a decision tree from training data, usually aiming to minimize in particular the size of the resulting tree.
  In contrast, little is known about the complexity of the underlying computational problem of computing a minimum-size tree for the given training data.
  We study this problem with respect to the number~$d$ of dimensions of the feature space.
  We show that it can be solved in $O(n^{2d + 1}d)$ time, but under reasonable complexity-theoretic assumptions it is not possible to achieve \lowerbound\ running time, where $n$ is the number of training examples.
  The problem is solvable in $(dR)^{O(dR)} \cdot n^{1+o(1)}$ time, if there are exactly two classes and $R$ is an upper bound on the number of tree leaves labeled with the first~class.
\end{abstract}

\section{Introduction}
A decision tree is a useful tool to classify and describe data~\cite{murthy_automatic_1998}.
It takes the feature space $\mathbb{R}^d$, recursively performs axis-parallel cuts to split the space into two subspaces, and then assigns class labels based on the resulting partition (see Figure~\ref{fig:example}).
Because of their simplicity, decision trees are particularly attractive as interpretable models of the underlying data~\cite{molnar_interpretable_2020}.
In this context, small trees are preferable, i.e., trees that have a small number of nodes, or in other words, perform a small number of cuts~\cite{moshkovitz_explainable_2020}.
Such trees are also desired in the context of classification, because it is thought that minimizing the number of nodes reduces the chances of overfitting~\cite{fayyad_what_1990}.

\looseness=-1
In the learning phase, we are given a finite set of examples $E \subseteq \mathbb{R}^d$ labeled with classes and we want to find a decision tree that optimizes certain performance criteria and that is consistent with $E$, that is, the classes assigned by the tree agree with the class labels of the examples.
Among other criteria, the number of nodes is often minimized for the above-mentioned reasons.
There is a plethora of implementations for learning decision trees (e.g., \cite{breiman_classification_1984,bessiere_minimising_2009,narodytska_learning_2018,schidler_satbased_2021,DBLP:conf/nips/HuRS19}).
One of the heuristics herein is among the Top 10 Algorithms of Data Mining chosen by the ICDM~\cite{wu_top_2008,steinberg_cart_2009} and several implementations are based on exact algorithms minimizing the size of the produced trees.
Despite this, our knowledge of the computational complexity of learning (minimum-node) decision trees is limited: Several classical results show NP-hardness~\cite{hyafil_constructing_1976,GMOS95} (see also the survey by \citet{murthy_automatic_1998}) and we know that even if we require parameters such as the number of nodes of the tree, or the number of different feature values, to be small, we still cannot achieve efficient algorithms in terms of upper bounds on the running time~\cite{DBLP:conf/aaai/OrdyniakS21}.

\looseness = -1
In this paper, we study the influence of the number~$d$ of dimensions of the feature space on the complexity of learning small decision trees.
This problem can be phrased as the decision problem \dtslong~(\dts): The input is a tuple $(E, \lambda, s)$ consisting of a set $E \subseteq \mathbb{R}^d$ of examples, a class labeling $\lambda \colon E \to \{\lpos, \lneg\}$, and an integer $s$, and we want to decide whether there is a decision tree for $(E, \lambda)$ of size at most~$s$.
Herein, a binary tree $T$ is decision tree for $(E, \lambda)$ if the labeled partition of $\mathbb{R}^d$ associated with $T$ agrees with the labels $\lambda$ of $E$, see \cref{sec:prelims} for a precise definition.

We provide three main results.
First, we show that \dts\ can be solved in $O(n^{2d + 1}d)$ time (\cref{thm:xp-time}), where $n$ is the number of input examples.
In other words, for fixed number of dimensions, \dts\ is polynomial-time solvable.
Contrast this with the variant where, instead of axis-parallel cuts, we allow linear cuts of arbitrary slopes.
The problem then becomes NP-hard already for~$d=3$~\cite{GMOS95}.

Second, complementing the first result, we show that the dependency on $d$ in the exponent cannot be substantially reduced.
More precisely, we show that a running time of \lowerbound\ would contradict well-known reasonable complexity-theoretic assumptions (\cref{thm:whardness}).
This implies that the running time of algorithms for \dts\ has to scale exponentially with~$d$.
In other words, any provably efficient algorithm for \dts\ has to exploit other properties of the input or desired solution.

\looseness = -1
Our third pair of results determines more closely what parameters influence the combinatorial explosion, offering two tractability results.
A crucial property of a construction that we use in \cref{thm:whardness} is that the size of the optimal decision tree is unbounded.
Informally, this result thus shows intractability only in situations where the smallest decision tree for our input data is rather large.
This may be the case for practical data (partially overlapping classes of Gaussian-distributed data), but it begs the question whether we can find particularly small decision trees provably efficiently if the data allow for it.
As Ordyniak and Szeider showed, without further restrictions, this is not possible as \dts\ is W[2]-hard with respect to the solution size~$s$~\cite{DBLP:conf/aaai/OrdyniakS21}.
In contrast, we show that in the small-dimension regime we do obtain a prospect for an efficient algorithm with running time $O((s^3d)^{s} \cdot n^{1 + o(1)})$ (\cref{thm:improved-binary-search}).
An intermediate result towards this is inspired by and improves upon an algorithm by \citet{DBLP:conf/aaai/OrdyniakS21} for determining a smallest decision tree that cuts only a given set of features; we decrease the running time from $2^{O(s^2)} \cdot \poly(n)$ to $2^{O(s \log s)} \cdot n^{1 + o(1)}$ for our purpose.

Finally, we show that in the tractability result with respect to $s$ and $d$, the size $s$ can be replaced by an a priori even smaller parameter:
Let $R$ be an upper bound on the number of leaves in the decision tree labeled with any one class.
Equivalently, $R$ is an upper bound on the number of parts in the partition induced by the tree that contain examples of the first class, or that only contain examples of the second class, whichever number is smaller.
Then, \dts\ is solvable in $(dR)^{O(dR)} \cdot n^{1+o(1)}$ time (\cref{thm:fpt-dim-leaves}).
While mostly interesting from a theoretical perspective, as \dts\ is NP-hard even for $R = 1$ in the unbounded-dimension regime \cite{DBLP:conf/aaai/OrdyniakS21}, we believe that restricting the desired decision tree to a small number $R$ of leaves labeled with one class has a reasonable practical motivation: In some situations the distribution of the data may not allow for particularly small decision trees, hampering interpretability.
In this case, it may be useful to look for trees in which one class labels few leaves. %

Summarizing, while $n^{O(d)}$-time algorithms for \dts\ are achievable, they cannot be substantially improved in general, but when restricting to small solution sizes or a class with small number of leaves, there are prospects for efficient algorithms.

%
%
%

%
%
%
%
%
%
%
%
%
%
%

%

%

%
%
%
%
%
%
%
%
%
%

%
%
%
%
%
%
%
%
%

\section{Preliminaries}\label{sec:prelims}

For $n \in \mathbb{N}$ we use $[n] := \{1, 2, \ldots, n\}$.
For a vector $x \in \mathbb{R}^d$ we denote by $x[i]$ the $i$th entry of~$x$.

Let $E \subseteq \mathbb{R}^d$ and $\lambda \colon E \to \{\lpos, \lneg\}$. Furthermore, let the domain~$D_i := \{ x[i] \mid x \in E\}$ of~$E$ consist of all distinct coordinate values for dimension~$i$ occurring in examples of~$E$.
We aim to define what a decision tree for $(E, \lambda)$ is.
Let $T$ be a rooted and ordered binary tree and let $\dimn \colon V(T) \to [d]$ and $\thr \colon V(T) \to \mathbb{R}$ be labelings of each inner node $t \in V(T)$ by a \emph{dimension} $\dimn(t) \in [d]$ and a \emph{threshold} $\thr(t) \in \mathbb{R}$.
For each inner node $t$ of $T$ there is a left and a right child of $t$, labeled by $\leq$ and $>$, respectively; see Figure~\ref{fig:example} for an illustration.

\begin{figure}[t]
    \centering
    \includegraphics{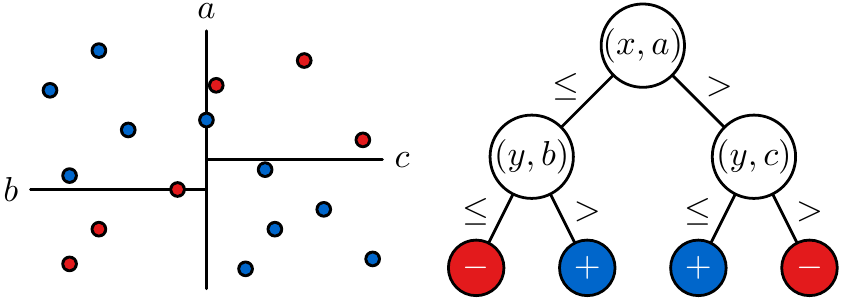}
    \caption{Left: An instance $(E,\lambda)$ of \dts\ with two dimensions $x$ (horizontal) and $y$ (vertical).
      Points are examples and the $\lpos$ and $\lneg$ color represents their labels assigned by $\lambda$.
      Right: A minimum decision tree $T$ for $(E,\lambda)$.
      Each internal node $t\in T$ is labeled by $(\dimn(t),\thr(t))$.}
    \label{fig:example}
\end{figure}

\looseness = -1
We use~$E[f_i \leq t] = \{ x\in E \mid x[i] \leq t\}$ and~$E[f_i > t] = \{ x\in E \mid x[i] > t\}$ to denote the set examples of~$E$ whose $i$th dimension is less or equal and strictly greater than some threshold~$t$, respectively.
Each node $t \in V(T)$, including the leaves, defines a subset $E[T, t] \subseteq E$ as follows.  For the root~$t$ of~$T$, we define~$E[T,t] := E$.  For each non-root node~$t$ let~$p$ denote the parent of~$t$.  We then define~$E[T,t] := E[T,p] \cap E[f_{\dimn(p)} \le \thr(p)]$ if $t$ is the left child of~$p$ and~$E[T,t] := E[T,p] \cap E[f_{\dimn(p)} > \thr(p)]$ if $t$ is the right child of~$p$.
If the tree $T$ is clear from the context, we simplify $E[T, t]$ to $E[t]$.

Now $T$ and the associated labelings are a \emph{decision tree} for $(E, \lambda)$ if for each leaf~$\ell$ of $T$ we have that all examples in $E[\ell]$ have the same label under $\lambda$.
Below we will sometimes omit explicit reference to the labelings of the nodes and edges of $T$ and simply say that $T$ is a decision tree for $(E, \lambda)$.
The \emph{size} of $T$ is the number of its inner nodes.
We conveniently call the nodes of $T$ and their associated labels \emph{cuts}.
The problem \dtslong~(\dts) is defined as in the introduction. For most of our results, the number of classes of the labeling~$\lambda$ does not have to be restricted to two. We therefore also introduce \kdts\ as the generalization of \dts\ in which $\lambda \colon E \to [k]$.

\looseness=-1
Our analysis is within the framework of parameterized complexity~\cite{gottlob_fixedparameter_2002}.
Let $L \subseteq \Sigma^{*}$ be a computational problem specified over some alphabet $\Sigma$ and let $p \colon \Sigma^{*} \to \mathbb{N}$ a parameter, that is, $p$ assigns to each instance of $L$ an integer parameter value (which we simply denote by $p$ if the instance is clear from the context).
We say that $L$ is \emph{fixed-parameter tractable} (FPT) with respect to $p$ if it can be decided in $f(p) \cdot \poly(n)$ time where $n$ is the input encoding length.
A complement to fixed-parameter tractability is W[$t$]-hardness, $t \geq 1$; if problem $L$ is W[$t$]-hard with respect to $p$ then it is thought to not be fixed-parameter tractable; see \cite{flum_parameterized_2006,Nie06,CyFoKoLoMaPiPiSa2015,downey_fundamentals_2013} for details.
The Exponential Time Hypothesis (ETH) states that \textsc{3SAT} on $n$-variable formulas cannot be solved in $2^{o(n)}$~time, see refs.~\cite{impagliazzo_complexity_2001,impagliazzo_which_2001} for details.

\section{Algorithms}

\looseness=-1
We now present our algorithmic results, that is, that \kdts\ is polynomial-time solvable for constant number of dimensions, an improved algorithm for the case where we are restricted to a given set of dimensions to cut, a fixed-parameter algorithm for $d + s$ and one for~$d + R$.
For simplicity, we give our algorithms for the decision problem ($k$-)\dts, but with standard techniques they are easily adaptable to also producing the corresponding tree if it exists and to the corresponding size-minimization problem.

For the first result we use the order of the examples~$E$ in each dimension. In a dimension~$i$, let $c_1, c_2, \ldots, c_n$ be the coordinate values in $D_i$ in ascending order. We use the set~$S_i$ of \emph{splits}, that is, cuts that each partition $E$ in a combinatorially distinct way%
. Specifically, we can define the set $S_i := \{(c_j + c_{j+1})/2 \mid j\in[n-1]\} \cup \{-\infty,\infty\}$. When there are at most $D_{\max}$ different values used in each dimension, we have the following bound: $|S_i| = |D_i| -1 +2 \leq D_{\max}+1 = O(n)$.

\begin{theorem}\label{thm:xp-time}
    \kdts\ is solvable in $O(D_{\max}^{2d}dn) = O(n^{2d+1}d)$ time. 
\end{theorem}
\begin{proof}
    Observe that, by adjusting the thresholds, a minimum decision tree for~$E$ can be modified such that for each node~$t$ it is~$\thr(t) \in S_{\dimn(t)}$.  Let~$S = S_1 \times S_2 \times \ldots \times S_d$. %
    
    \looseness = -1
    We do dynamic programming on hyperrectangles defined by two points: for $u,v\in S$, let $\boxx(u,v)$ be an axis-aligned hyperrectangle with $u$ and $v$ as antipodal corners.  We require that each coordinate of $u$ is smaller than the respective coordinate of $v$. 
    We are interested in computing a minimum decision tree $T$ for the examples in $E \cap \boxx(u,v)$.  This solution can be found by cutting $\boxx(u,v)$ with an axis-aligned hyperplane~$H$, and combining the minimum-size decision trees~$T_1,T_2$ for the examples on either side of the hyperplane. Since hyperplane~$H$ is defined by a dimension~$i$ and threshold~$t \in S_i$, we can find two new points $u'$ and $v'$, whose coordinates coincide with $u$ and $v$, respectively, but in dimension $i$ they have coordinate $t$.  Note that, since no example has coordinate~$t$ in dimension~$i$, $E \cap \boxx(u,v')$ and~$E \cap \boxx(u',v)$ partition the examples in~$E \cap \boxx(u,v)$; see Figure~\ref{fig:box-cut}.
    
    \begin{figure}[t]
        \centering
        \includegraphics{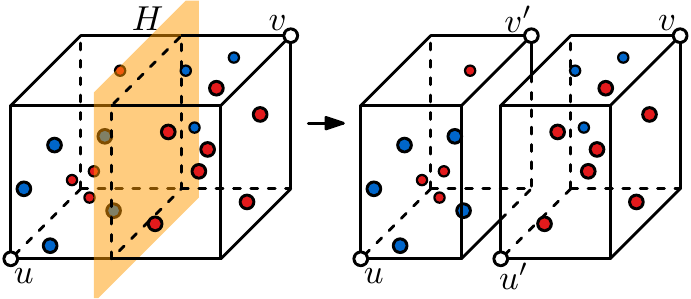}
        \caption{3D instance where hyperplane~$H$ cuts $\boxx(u,v)$ into $\boxx(u,v')$ and $\boxx(u',v)$, and splits examples $E \cap \boxx(u,v)$ into $E \cap \boxx(u,v')$ and $E\cap \boxx(u',v)$.}
        \label{fig:box-cut}
    \end{figure}
    
    We can therefore use hyperplane~$H$ to define the root node~$(i,t)$ of decision tree~$T$, with $T_1$ and $T_2$ as subtrees. These subtrees are found by recursively computing a solution for $E \cap \boxx(u,v')$ and $E\cap \boxx(u',v)$.  Since there are $i$ dimensions, and each of them admits at most~$D_{\max}+1$ distinct hyperplanes that we can use to split, an optimal solution for~$E \cap \boxx(u,v)$ can be computed by trying all $O(D_{\max}\cdot d)$ hyperplanes defined by distinct splits in each dimension.
    
    Formally, let~$\mathcal{T}$ be a dynamic programming table, in which for each $u, v \in S$ with $u \leq v$, entry~$\mathcal{T}[u,v]$ holds the minimum size of a decision tree for $E \cap \boxx(u,v)$.
    Define the \emph{volume} of $\boxx(u, v)$ as the number of grid points contained within it, that is, $|S \cap \boxx(u, v)|$.
    We fill the table in a bottom-up fashion from smaller-volume boxes to larger-volume boxes.
    We iterate over the range of possible volumes $k$ from one to $|S|$ and we compute $\mathcal{T}[u, v]$ for all $u, v \in S$ such that $\boxx(u, v)$ has volume~$k$.
    For each such $u, v$, we first take $O(nd)$ time to check whether the examples in $\boxx(u, v)$ all have the same label.
    If so, then we put $\mathcal{T}[u,v]=0$.
    Otherwise, $\mathcal{T}[u,v]$ is defined by the following recurrence: 
    $$\mathcal{T}[u,v] = \min_{i\in [d]}\min_{\sigma \in S_i} \mathcal{T}[u,v_i(\sigma)] + \mathcal{T}[u_i(\sigma),v] + 1.$$
    Thus, for each coordinate $i$ of $u$ and $v$, we use the values $S_i$ as options.
    Note that each considered table entry on the right-hand side corresponds to a box of strictly smaller volume than the box on the left-hand side. 
    To see that the recurrence is correct, observe that a minimum-size decision tree~$T$ for $E \cap \boxx(u, v)$ requires at least one cut.
    Consider the cut $t$ at the root of $T$, shifted to a closest split if necessary.
    At some point while taking the minimum we have $i = \dimn(t)$ and $\sigma = \thr(t)$.
    The subtrees of $T$ rooted at the children of $t$ are decision trees for $E \cap \boxx(u,v_i(\sigma))$ and $E \cap \boxx(u_i(\sigma),v)$, respectively.
    Thus, the left-hand side upper bounds the right-hand side.
    For the other direction, observe that combining any two decision trees for $E \cap \boxx(u,v_i(\sigma))$ and $E \cap \boxx(u_i(\sigma),v)$ with a cut at $\sigma$ in dimension~$i$ gives a decision tree for $E \cap \boxx(u, v)$, as required.
    Finally, the size of the minimum-size decision tree for $E$ can be found in $\mathcal{T}[u^*, v^*]$, where all coordinates of $u^*$ and $v^*$ are $-\infty$ and $\infty$, respectively.
    
    As to the running time, table~$\mathcal{T}$ has $O(D_{\max}^{2d})$ entries, each of which each takes $O(nd)$ time to fill: checking for consistency and then trying each distinct split and looking up the size of the respective minimum-size subtrees.
    We proceed by increasing the domain in each dimension by one split at a time, one coordinate at a time.
    Thus, the total running time adds up to $O(D_{\max}^{2d}nd) = O(n^{2d+1}d)$.
\end{proof}

\looseness=-1
Before we elaborate on our next result, we first show how to improve Theorem~4 in~\cite{DBLP:conf/aaai/OrdyniakS21}, which shows that, given an instance $(E, \lambda, s)$ of \dts, and given a subset~$\mathcal{D}$ of the dimensions, it is possible to compute in $2^{O(s^2)}|E|^{1+o(1)}\log{|E|}$ time the smallest decision tree among all decision trees of size at most~$s$ (if they exists), that use exactly the given subset~$\mathcal{D}$ in their cuts --- none can be left out.
The main idea is to first enumerate the structure of all possible decision trees of size~$s$, before finding thresholds that work for the instance.
Instead of enumerating all possible decision trees and finding the right thresholds afterwards, we interleave the two processes, see Algorithm~\ref{alg:improved-binary-search}. Furthermore, we no longer take as input a subset~$\mathcal{D}$ of dimensions, which will be used for labeling internal nodes in the decision tree, but we simply bound the number~$d$ of dimensions in the instance.
By iterating over a subset~$\mathcal{D}$ instead of all $d$ dimensions, Algorithm~\ref{alg:improved-binary-search} can be adapted towards the initial setting.

\renewcommand{\algorithmicrequire}{\textbf{Input:}}
\renewcommand{\algorithmicensure}{\textbf{Output:}}

\begin{algorithm}[t]
  \caption{\textsc{SmallestDecisionTree}$(E, d,s)$}\label{alg:improved-binary-search}
  \begin{algorithmic}[1]
  \REQUIRE Example set~$E$ and numbers~$s$ and~$d$
  \ENSURE Decision tree~$T$ for $E$ with at most $s$ internal nodes using $d$ dimensions to label internal nodes.
  
  \smallskip 
\STATE Set $sdt$ to $nil$, with $|nil| = \infty$
\IF{$s = 0$}
    \IF{$E$ is uniform}\label{alg:uniform}
        \RETURN $leaf$, with $|leaf| = 0$
    \ELSE
        \RETURN $nil$, with $|nil| = \infty$
    \ENDIF
\ENDIF
\FOR{$i = 1$ to $d$}
    \FOR{$j = 0$ to $s-1$}
        \STATE $t = \textsc{BinarySearch}(E, i, j)$
        \STATE $r = \textsc{SmallestDecisionTree}(E[f_i > t],d, s-j-1)$
        \STATE $l = \textsc{SmallestDecisionTree}(E[f_i \leq t], d,j)$
        \STATE $dt = (f_i = t) \cup (l, r)$, with $|dt| = |l| + |r| + 1$
        \IF{$|dt| < |sdt|$}\label{alg:smallest}
            \STATE $sdt = dt$
        \ENDIF
    \ENDFOR
\ENDFOR
\RETURN $sdt$
\end{algorithmic}
\end{algorithm}

\begin{algorithm}[t]
  \caption{\textsc{BinarySearch}$(E,i,j)$}\label{alg:binary-search}
  \begin{algorithmic}[1]
  \REQUIRE Example set~$E$ and numbers~$i$ and $j$
  \ENSURE Largest threshold $t$ for which $E[f_i \leq t]$ has a decision tree of size~$j$ 
  
  \smallskip 
\STATE Set $D$ to be an array containing $D_i$ in ascending order
\STATE Set $L = 0$, $R = |D_i|-1$, $b = 0$
\WHILE{$L \leq R$}
    \STATE $m = \lfloor(L+R) / 2\rfloor$
    \IF{$\textsc{SmallestDecisionTree}(E[f_i \leq D[m]],d,j) \neq nil$}
        \STATE $L = m+1$, $b = 1$
    \ELSE
        \STATE $R = m-1$, $b = 0$
    \ENDIF
\ENDWHILE
\IF{$b = 1$}
    \RETURN $D[m]$
\ENDIF
\RETURN $D[m-1]$, with $D[-1] = D[0]-1$
\end{algorithmic}
\end{algorithm}

\begin{restatable}{theorem}{improvedBinarySearch}\label{thm:improved-binary-search}
  \kdts\ is solvable in $O((s^3d)^{s}|E|^{1+o(1)})$
  time, and is FPT parameterized by $s+d$.%
\end{restatable}%
\begin{proof}
    Similar to the algorithm by Ordyniak and Szeider~\cite{DBLP:conf/aaai/OrdyniakS21}, we binary search for the largest threshold value~$t$ in dimension~$i$ for which the left subtree can still have a decision tree of size $j$ on the example set~$E[f_i \leq t]$. If we can find a decision tree of size at most $s-j-1$ for the remaining examples in the right subtree, then we have found a decision tree. However, if we cannot find such a decision tree for the right subtree, then we could not find a decision tree even for smaller threshold values of the root: there would be even more examples left for the right decision tree, which should still be of size at most $s-j-1$.
    
    However, instead of enumerating all trees and assignments of dimension labels to internal nodes, we loop over these options during the main procedure in Algorithm~\ref{alg:improved-binary-search}.
    We first prove that this still leads to an algorithm that correctly solves \kdts. When Algorithm~\ref{alg:improved-binary-search} returns a decision tree~$T$ (not $nil$), then this decision tree can have at most $s$ internal nodes and uses at most $d$ dimensions to make cuts. Thus $T$ is proof that there is a solution to \kdts\ parameterized by $s+d$. 
    
    \looseness = -1
    In the other direction, assume there is a solution $T^*$ for an instance $(E, \lambda, s)$ of \kdts.
    We show by induction over $s$ that \cref{alg:improved-binary-search} finds a decision tree of size at most~$s$.
    In the base case, where $s = 0$, the existence of $T^{*}$ shows that all examples in $E$ have the same label, which is also detected in Line~3 of \cref{alg:improved-binary-search}.
    Assume as an induction hypothesis that the statement holds for all instances of \kdts\ in which we seek for decision trees of size less than~$s$. 
    Consider the root $p^*$ of $T^*$ and let $s^*$ be the number of internal nodes of~$T^{*}$.
    Let $j^*$ be the number of internal nodes in the left subtree; then the number of internal nodes in the right subtree~$T^{*}_r$ of $T^{*}$ is necessarily $s^* - j^* - 1$.
    Since Algorithm~\ref{alg:improved-binary-search} loops through all dimensions and all ways of dividing the number of internal nodes over the left and right subtrees, there is one iteration of the loops in Line~7 and~8 in Algorithm~\ref{alg:improved-binary-search} where the chosen dimension is $\dimn(p^*)$, the chosen number of internal nodes for the left subtree equals $j^*$, and hence the right subtree should have at most $s^* - j^* - 1$ internal nodes.
    Decision tree $T^{*}_r$ witnesses that there is a decision tree for $E[f_{\dimn(p^{*})} \leq \thr(p^{*})]$ with at most $s^* - j^* - 1$ internal nodes.
    Thus, in Line~9 of \cref{alg:improved-binary-search}, \cref{alg:binary-search} will find a threshold $t \geq \thr(p^{*})$ because, by the induction hypothesis, the calls to \cref{alg:improved-binary-search} in \cref{alg:binary-search} in which $D[m] \leq \thr(p^{*})$ will return with a positive answer.
    Again by the induction hypothesis we will thus obtain decision trees of the corresponding sizes in Lines~10 and~11 of Algorithm~\ref{alg:improved-binary-search}.
    As we check in Line~\ref{alg:smallest} whether the found decision tree is smaller than the decision tree in another iteration, we will output a tree of size $s$ or smaller.
    Thus, Algorithm~\ref{alg:improved-binary-search} finds a decision tree that has at most~$s$ internal nodes.
    
    We now prove the running-time bound.
    Algorithms~\ref{alg:improved-binary-search} and~\ref{alg:binary-search} together build a recursion tree in which the nodes correspond to calls of Algorithm~\ref{alg:improved-binary-search}.
    In each node, corresponding to a call to Algorithm~\ref{alg:improved-binary-search} with some set $E$ and numbers $s$, $d$, there are at most $sd(2 + \log |E|) = sd\log(4|E|)$ recursive calls to Algorithm~\ref{alg:improved-binary-search}: For each iteration of the two loops in Algorithm~\ref{alg:improved-binary-search} there are two direct recursive calls and at most $\log |E|$ in Algorithm~\ref{alg:binary-search}.
    In each recursive call the parameter $s$ decreases by at least one.
    Hence, the overall size of the recursion tree is $O((sd\log(4|E|))^s)$.
    For each node $N$ of the recursion tree, we spend $O(|E|)$ time, as the main running time incurred by the call to Algorithm~\ref{alg:improved-binary-search} (corresponding to $N$) is in the uniformity check of $E$ (Line~\ref{alg:uniform} of \cref{alg:improved-binary-search}); the running time of the remaining bookkeeping tasks in Algorithm~\ref{alg:improved-binary-search} and~\ref{alg:binary-search} can be charged to the child nodes of $N$.%

    Finally, bounding $(\log(4|E|))^s$ uses the following well-known technique: We claim that $(\log m)^s = O(s^{2s} \cdot m^{1/s})$.
    To see this, observe that the claim is trivial for $m < s^{2s}$ (as then $(\log m)^s < (2s\log s)^s = s^s(2 \log s)^s \leq s^{2s}$).
    Otherwise, if $m \geq s^{2s}$, then we have $\log m \leq s^2 m^{1/s^{2}}$ because this holds for $m = s^{2s}$ (observe that dividing both sides by $s$ yields $2\log s \leq s^{1+ 2/s}$ which clearly holds) and the derivative wrt.\ $m$ of the left-hand side is at most as large as the derivative (wrt.\ $m$) of the right-hand side, that is, $1/(\ln(2) m) < m^{1/s^{2} - 1}$.
    Thus, in this case $(\log m)^s \leq s^{2s}m^{1/s}$, showing the claim.
    Substituting $4|E|$ for $m$ in $(\log m)^s = O(s^{2s} \cdot m^{1/s})$ we get the overall running time bound of $O((sd)^s \cdot s^{2s} \cdot |E|^{1 + 1/s}) = O((s^{3}d)^s \cdot |E|^{1 + o(1)})$.
\end{proof}

The strategy employed by Ordyniak and Szeider to solve \dts\ is to first find a subset~$\mathcal{D}$ of dimensions which should be cut to find a smallest decision tree~\cite{DBLP:conf/aaai/OrdyniakS21}.
Once the set~$\mathcal{D}$ has been determined, they use their Theorem~4 to find a smallest decision tree.
In our case, Algorithm~\ref{alg:improved-binary-search} works directly towards the final goal, both selecting dimensions to cut and finding a smallest decision tree at the same time.
We can adapt the algorithm to Ordyniak and Szeider's setting of cutting only within a specified set of dimensions by simply restricting the the dimensions to select in Line~7 in Algorithm~\ref{alg:improved-binary-search}, obtaining a more efficient running time.

\begin{corollary}
  Given a subset~$\mathcal{D}$ of dimensions, with $|\mathcal{D}| \leq s$, where we are allowed to cut only (a subset of) dimensions in~$\mathcal{D}$, \dts\ is solvable in $2^{O(s(\log{s}))} |E|^{1+o(1)}$ time.
\end{corollary}

We now consider the parameter~$R$.
For simplicity, we will in the following assume that $R$ restricts the number of red leaves in a decision tree, and we assume that there are at least as many blue leaves as red leaves.
Observe that this is without loss of generality, because to solve the general case where $R$ bounds the smaller one of the number of red leaves or the number of blue leaves we may simply try both options.
Below, we call an internal node that has red leaves in both subtrees an \emph{essential} node.

\begin{lemma}\label{lem:r-1-branching}
    A minimum-size decision tree $T$ with $R$ red leaves has at most $R-1$ essential nodes.%
\end{lemma}
\begin{proof}
  \looseness=-1
    Consider the subtree~$T'$, whose root is the essential node closest to the root of~$T$. There is only one such node, as either the root of~$T$ is essential, or one of its subtrees contains only blue leaves. %
    Remove from~$T'$ all maximal rooted subtrees that contain only blue leaves (possibly consisting of only a single blue leaf).  The resulting tree contains degree-2 nodes, which must be non-essential.
    Contracting the degree-2 nodes, we again obtain a binary tree~$T^*$.
    Note that the internal nodes of $T^{*}$ one-to-one correspond to the essential nodes of~$T$.
    Tree~$T^*$ still has all $R$ red leaves of~$T$ (and no blue ones), and thus consists of at most $R-1$ internal nodes.%
\end{proof}

\begin{lemma}\label{lem:2d-non-branching}
    In a minimum-size decision tree $T$ with $R$ red leaves, each root-to-leaf path has at most $2d$ consecutive non-essential nodes, where $d$ is the number of dimensions.
\end{lemma}
\begin{proof}
    Assume for a contradiction that a minimum-size decision tree~$T$ exists for example set~$E$ that has $2d+1$ consecutive non-essential nodes on a root-to-leaf path~$p$. Let $n_0$ be the essential node that is the child of the $2d+1$-th consecutive non-essential node in~$p$. Since there are $2d+1$ non-essential nodes, there must be three nodes $n_1 = (i,t_1)$, $n_2 = (j,t_2)$, and $n_3 = (l, t_3)$ such that $i=j=l$. Additionally, at least two of those nodes have a blue leaf as their right (or left) child. Assume w.l.o.g. that $n_1$ and $n_2$ have a blue leaf as their right child, $t_1<t_2$, and either $t_3<t_1$ or $t_2<t_3$. Thus, $n_1$ is closer to $n_0$ than $n_2$ on the root-to-leaf path~$p$ (see Figure~\ref{fig:path-of-blue}).
    
    \begin{figure}[t]
        \centering
        \begin{minipage}[t][200pt][b]{0.5\textwidth}
            \includegraphics{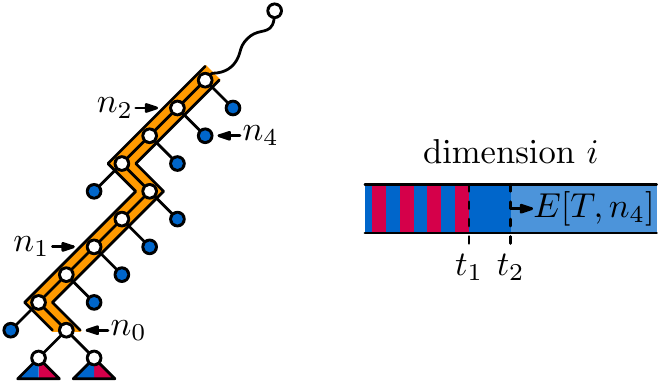}
            \caption{Construction of Lemma~\ref{lem:2d-non-branching}: path~$p$ in yellow, and nodes $n_1$ and $n_2$ respectively have thresholds $t_1$ and $t_2$ in dimension $i$.}
            \label{fig:path-of-blue}
        \end{minipage}
        \qquad
        \begin{minipage}[t][200pt][b]{0.35\textwidth}
            \centering       
            \includegraphics{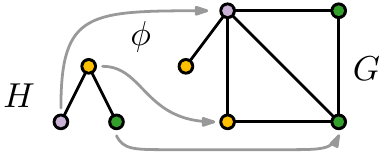}
            \caption{A subgraph isomorphism~$\phi$ (in gray) is a mapping from the vertices of $H$ to vertices of $G$.}
            \label{fig:isomorphism}
        \end{minipage}
    \end{figure}
    
    Now consider a decision tree~$T^*$ that is identical to $T$, except it does not contain $n_2$, nor the blue leaf~$n_4$ attached at $n_2$ (as its right child) --- the parent of $n_2$ is directly connected to the internal node that is the left child of $n_2$. The node $n_0^*$ in $T^*$, corresponding to $n_0$ in $T$, is unaffected by this change, meaning that $E[T,n_0]=E[T^*,n_0^*]$, for the following reason. All blue examples in $E[T,n_4]$ which follow the root-to-leaf path to $n_1^*$ in $T^*$ (corresponding to $n_1$ in $T$), will not reach $n_0^*$, since in dimension $i=j$ each such example $e$ has a coordinate $c_e$, for which holds that $t_1<t_2<c_e$. Thus, such examples belong in the leaf node connected to $n_1^*$.
    As a result, $T^*$ must be a smaller decision tree for~$E$ than $T$, contradicting the assumption that~$T$ is a minimum-size decision tree.
\end{proof}
\looseness=-1
Lemmas~\ref{lem:r-1-branching} and~\ref{lem:2d-non-branching} together show that a minimum size decision tree has at most $2d$ non-essential nodes before each essential node and each red leaf. Thus such a tree has at most $2d(2R-1)$ internal nodes. We can therefore apply Theorem~\ref{thm:improved-binary-search} to prove that \dts\ is FPT with $d$ and $R$ as parameters.

\begin{theorem}\label{thm:fpt-dim-leaves}
\dts\ is solvable in $O((s^3d)^{s}|E|^{1+o(1)})$ time, with $s=2d(2R-1)$, and hence \dts\ is FPT parameterized by $d+R$.
\end{theorem}

\section{Running-time lower bound%
}
We now prove the following lower bound for computing decision trees.
\begin{theorem}\label{thm:whardness}
  \dtslong\ (\dts) is W[1]-hard with respect to the number~$d$ of dimensions.
  Assuming the \ethlong\ there is no algorithm solving \dts\ in time $f(d)n^{o(d/\log d)}$ where $n$ is the input size and $f$ is a computable function.
\end{theorem}
\looseness=-1
The remainder of this section is devoted to the proof of \cref{thm:whardness}.
Below, for a graph $G$ we use $V(G)$ to denote its vertex set and $E(G)$ for its edge set.
We give a reduction from the \subisolong~(\subiso) problem.
Its input consists of a graph $G$ with a proper $K$-coloring~$\col \colon V(G) \to [K]$, such that no two adjacent vertices share a color, and a graph~$H$ with vertex set $[K]$ that has no isolated vertices.   
The question is whether $H$ is isomorphic to a subgraph of~$G$ that respects the colors, i.e., whether there exists a mapping~$\phi \colon V(H) \to V(G)$ such that (i) each vertex of $H$ is mapped to a vertex of~$G$ with its color, i.e., $\col(\phi(c)) = c$ for each~$c \in V(H)$ and (ii) for each edge $\{c,c'\}$ in~$H$, $\{\phi(c),\phi(c')\} \in E(G)$ is an edge of~$G$. See Figure~\ref{fig:isomorphism} for an example of such a mapping.
In that case, we also say that $\phi$ is a \emph{subgraph isomorphism} from $H$ into~$G$.
In the following, we let $m_G = |E(G)|$, $n_H = |V(H)|$ and $m_H = |E(H)|$.
Observe that $n_H \leq 2m_H$ since $H$ has no isolated vertices.
For each color $k \in [K]$, we denote by $V_k = \{v \in V(G) \mid \col(v) = k\}$ the vertices of~$G$ with color~$k$.
We assume without loss of generality that there is $n \in \mathbb{N}$ such that for all $k \in [K]$ we have $|V_k| = n$ (otherwise add additional isolated vertices to~$G$ as needed) and that if there is no edge in $H$ between two vertices $u, v \in V(H)$, then there are no edges between $V_u$ and $V_v$ in~$G$.

\looseness=-1
Since \subiso\ contains the \textsc{Multicolored Clique} problem~\cite{fellows_parameterized_2009} as a special case, \subiso\ is W[1]\nobreakdash-hard with respect to~$m_H$.
Moreover, Marx~\cite[Corollary 6.3]{marx_can_2010} observed that an $f(m_H)\cdot n^{o(m_H/\log m_H)}$-time algorithm for \subiso\ would contradict the \ethlong.
Our reduction will transfer this property to \dts\ parameterized by the number of dimensions.

\emph{Outline.} Given an instance $(G, H)$ of \subiso\ we now describe how to construct an equivalent instance $(E, \lambda)$ of \dts. 
Our construction consists of two types of gadgets.
First, for each edge $\{c,c'\}\in E(H)$, we use a two-dimensional \emph{edge-selection subspace} to model the choice for an edge $\{u,v\}\in E(G)$ with $\col(u)=c,\col(v)=c'$.
Second, for each vertex $c\in V(H)$, we use a one-dimensional \emph{vertex-verification subspace} to check whether the chosen edges with an endpoint of color $c$ consistently end in the same vertex $u\in V(G)$ with $\col(u)=c$.
Furthermore, we classify examples in our construction into two types: \emph{primary} and \emph{dummy} examples. We use primary examples to model vertices and edges in $G$, while dummy examples are used only to force certain cuts in the constructed instance.

\looseness=-1 We first describe the constructed instance $(E, \lambda)$ of \dts\ by giving the labeled point sets that we obtain when projecting the examples in $E$ to the edge-selection and vertex-verification subspaces.
Later, we define the examples in $E$ by giving the points they project to in each of the subspaces.
We specify labels for most of these points, and primary examples may project only to points with a matching label (red or blue), whereas dummy examples can project to any point.
In each subspace there will be several points that will be used by examples in order to achieve the correct behavior of each gadget.
On the other hand, most examples will play a role only in very few subspaces and the other dimensions shall not be relevant for them.
To achieve this property, we reserve in each subspace one unlabeled point (usually with the minimum or the maximum coordinate) that can be used by all examples that shall not be separated from each other in this specific subspace.
The vertex-verification subspaces have a second unlabeled point that can only be used by dummy vertices.
We call an example that projects to the unlabeled point of some subspace \emph{irrelevant} for this subspace and conversely, the subspace is irrelevant for this example.

We need some more tools to describe the points in the edge-selection and vertex-verification subspaces:
In one-dimensional subspaces, the precise coordinates of the points do not matter and we rather specify their order.
The main ingredient in the construction are pairs of a red and a blue point that need to be separated:
An \emph{$rb$-pair} is a pair~$(r,b)$ of points that are consecutive in the linear order with the red point preceding the blue point.
To avoid that $rb$-pairs interfere with each other, we separate them with forced cuts.
To achieve this, we use what we call dummy tuples.
A \emph{dummy tuple} consists of $2(m_G + 2)$ points (\emph{dummy points}) to which only dummy examples can project.
The first two are red, the second two are blue, and so on, and the last two are blue (without loss of generality, we assume that $m_G$ is even); see Figure~\ref{fig:vertex-selection-b}a.
Dummy tuples are placed between consecutive $rb$-pairs.
We later project dummy examples to the points in a dummy tuple so to ensure the following two properties.
First, only examples with colors matching the respective point in the tuple can project to such a point.
Second, these examples force $m_G + 3$ cuts in the corresponding subspace as follows.
A number of $m_G + 1$ cuts must be placed between each pair of equally colored and adjacent middle dummy points, since we ensure that the examples that project here differ only in the subspace of this dummy tuple and in no other subspace.
Additionally, two cuts must be placed between the outer dummy points and the adjacent $rb$-pair, since we ensure that the dummy examples that project to the outer points differ from an example projecting to the neighboring point in the $rb$-pair only in the subspace of the dummy tuple.

\looseness=-1
\emph{Edge-selection.}
We now describe a two-dimensional edge-selection subspace $S_e$ for an edge $e$ of $H$, see Figure~\ref{fig:vertex-selection-b}b for an illustration.
We refer to the two dimensions of $S_e$ by \emph{$x$ and $y$-dimension of $S_e$}.
Except for the unlabeled point, each point~$p$ has coordinates of the form $(c_p,c_p)$, and we therefore simply specify the linear point order.
Let~$e_1, e_{2}, \dots,e_j$ denote the edges of $G$ whose endpoints have the same colors as $e$.
For each edge~$e_i$, we place an $rb$-pair called \emph{$e_i$'s edge pair}.
Between any two edge pairs we put a dummy tuple.
We place an unlabeled point whose $x$-coordinate is smaller than that of any other point and whose $y$-coordinate is larger than that of any other point.
We allocate a budget of~$(j-1) \cdot (m_G + 4)$ cuts that shall be used for performing cuts in these two dimensions.
The idea is that, by using $(j-1) \cdot (m_G + 3)$ cuts, it is possible to have a cut between every edge pair and the dummy tuples adjacent to it ($2(j-1)$ cuts) and the inner pairs of each dummy tuple ($(j - 1) \cdot (m_G + 1)$ cuts), and the remaining $j-1$ cuts can then be used to cut all but one of the edge pairs, which corresponds to choosing the edge whose edge pair is not cut.
Conversely, for each edge~$e$ of~$H$, there is a decision tree of size~$(j-1) \cdot (m_G + 4)$ for the points in the subspace, for which only the example set of a single leaf contains both red and blue points, and it contains precisely the edge pair of the edge~$e_i$ of~$G$ and the unlabeled point; see Figure~\ref{fig:vertex-selection-b}b.

\begin{figure}
  \centering       
  \includegraphics{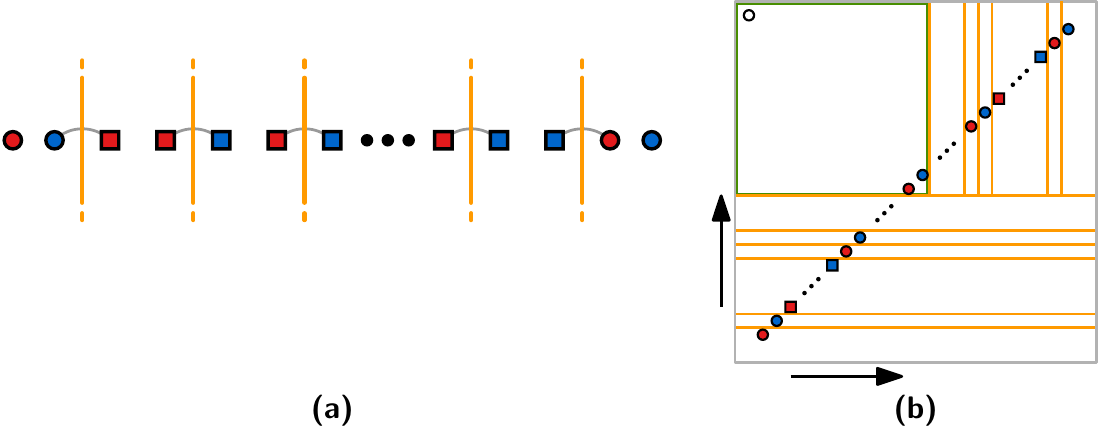}
  \caption{\textbf{\textsf{(a)}} Two $rb$-pairs (disks) and a dummy tuple (squares) between them, along with cuts forced by the dummy tuple. Examples projecting to connected points differ only in this dimension. \textbf{\textsf{(b)}} An edge-selection subspace. Consecutive red-blue pairs are shown as disks, while dummy examples are squares. The unlabeled point is shown in white. The yellow cuts show how one pair will be left unseparated, and cuts can be placed such that this pair is not separated from the unlabeled point either.}
  \label{fig:vertex-selection-b}
\end{figure}
 
\looseness=-1
\emph{Vertex-verification.}
Next, we describe a vertex-verification subspace for a vertex~$v$ of $H$.
The vertex-verification subspace of~$v$ is one-dimensional, and we again describe its projection by giving the order of the labeled points.
We first place a left unlabeled point, then one $rb$-pair, called a \emph{vertex pair}, for each vertex of $G$ whose color is $v$, and then a right unlabeled point; see Figure~\ref{fig:construction-example}.
We allocate a budget of a single cut for each vertex-verification space.
This budget allows to place a cut that separates one vertex pair as well as the left unlabeled and the right unlabeled point.
The idea is that all but the vertex pair that corresponds to the vertex of~$G$ that has been selected shall be separated by cuts in the edge-selection subspaces, and therefore a single cut suffices in the vertex-verification subspace.

\emph{Synthesis and examples.}
We now describe the instance~$(E, \lambda, s)$ of \dts\ that we construct for a given instance $(G,H)$ of {\sc Partitioned Subgraph Isomorphism}.
Our examples are elements of a space that contains $m_H$ two-dimensional edge-selection subspaces and $n_H$ one-dimensional vertex-verification subspaces, i.e., our points are in dimension~$d=2m_H+n_H$.
According to the budgets of cuts for the subspaces given above, we put the upper bound $s$ on the size of the desired decision tree to be $s = (m_G + 4) \cdot (m_G - m_H)+n_H$.%

\looseness=-1
Our construction contains two red and two blue primary examples for each edge of $G$.
For each edge $e=\{u,v\}$ of $G$, let~$S_{uv}$ denote the edge-selection subspace corresponding to the edge $\{\col(u),\col(v)\}$ of $H$ and let~$S_u$ and $S_v$ denote the vertex-verification subspaces corresponding to~$\col(u)$ and~$\col(v)$, respectively.
We create two primary example pairs~$U$ and~$V$, each consisting of a red and a blue example, which project to the vertex pair corresponding to~$u$ and~$v$ in~$S_u$ and in~$S_v$, respectively.
They both project to the edge pair of $e$ in~$S_{uv}$.
In all other dimensions, these pairs project to the (right) unlabeled point.  This finishes the construction of our primary examples. 

We now describe the dummy examples.
We create for each dummy tuple~$D$ contained in an edge-selection subspace~$S$ a number of $2(m_G + 1) + 1$ pairs of examples $L_1, L_2, \ldots, L_{m_G + 1}, R_1, R_2, \ldots, R_{m_G + 1}, P$ that each consist of a red and a blue dummy example.
In the subspace~$S$, the pairs~$L_i$ and~$R_i$ project to pairs of adjacent red and blue points in the middle of $D$.
In all other edge-selection subspaces, they project to the unlabeled point.
In each vertex-verification subspace, the pairs $L_i$ and $R_i$ project to the left and the right unlabeled point, respectively.
The red example of the pair~$P$ projects in~$S$ to the outer red point of $D$, and it coincides in all other subspaces with some fixed blue primary example~$b$ that projects to the blue point preceding the outer red point of $D$ in~$S$.
Likewise, the blue example of~$P$ projects in~$S$ to the outer blue point of~$D$, and it coincides in all other subspaces with some fixed red primary example~$r$ that projects to the red point succeeding the outer red point of $D$ in~$S$.
Observe that the examples of~$P$ can be separated from $r$ and~$b$ only in subspace~$S$, and therefore force the presence of two cuts.
Similarly, each of $L$ and~$R$ and the remaining examples can be separated only in~$S$.

Finally, we create for each vertex-verification subspace~$S$ one dummy pair~$U$ whose red and blue examples project to the left and to the right unlabeled point in~$S$, respectively.
In all other dimensions, they project to the (right) unlabeled point.
A key technical point is that the pair~$U$ enforces at least one cut in~$S$.
However, if we separate~$U$ by some cut~$C$ before separating both of the pairs~$L$ and~$R$ of some dummy tuple~$D$, then~$L$ and~$R$ end up on different sides of~$C$, thereby increasing the necessary number of cuts in the edge-selection subspace of~$D$.

\begin{figure}
  \centering
  \includegraphics{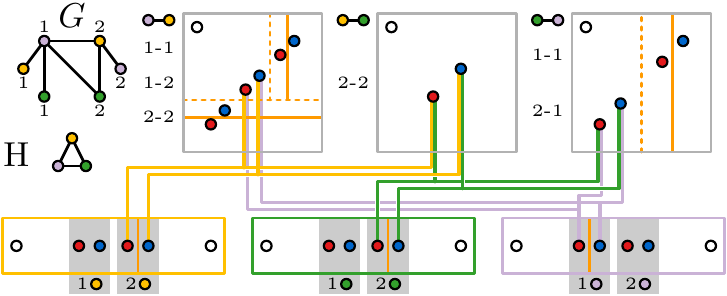}
  \caption{An instance $(E,\lambda)$ for graphs~$G$ and~$H$. The gray squares show edge-selection subspaces for the each color combination. The combinations and indices of colored vertices in $G$ are shown left of each subspace. The colored rectangles show vertex-verification dimensions for the corresponding color, with vertex pairs per index having gray backgrounds. Unlabeled points are shown in white. Edges connecting points in different subspaces correspond to examples of vertices in~$H$, and show the projection of those examples in the respective subspaces. The yellow cuts correspond to a solution to \dts in $(G,H,\col)$, with dashed cuts indicating multiple cuts through (omitted) dummy tuples.}
  \label{fig:construction-example}
\end{figure}

\begin{lemma}
  The instance $(G,H,\col)$ of \subiso\ is a yes-instance, if and only if the instance~$(E,\lambda,s)$ of \dts\ is a yes-instance.
\end{lemma}
\begin{proof}
  $\Rightarrow$: If $(G,H,\col)$ is a yes-instance for \subiso, then there is a subgraph isomorphism~$\phi$ from $H$ to $G$.
  We construct a decision tree for~$(E,\lambda,s)$ as follows.
  The decision tree's inner nodes form a path and thus we specify the sequence of cuts that we make in the nodes of the path.
  We start by placing cuts in the edge-selection subspace~$S_e$ for each edge~$e\in E(H)$.
  Denote the edges of~$G$ whose endpoints match the colors of~$e$ by~$e_1, e_2, \dots,e_n$ as in the description of~$S_e$.
  Let~$i$ be such that~$\phi(e) = e_i$.
  We first make the following cuts in the $y$-dimension of $S_e$ in ascending order.
  Cut between the edge pairs of~$e_1, e_2, \dots,e_{i-1}$, and make all cuts prescribed by dummy tuples succeeding~$e_1,\dots,e_{i-1}$.
  Then we make the following cuts in the $x$-dimension of $S_e$ in descending order.
  Cut between the edge pairs of~$e_n, e_{n - 1}, \dots,e_{i+1}$, and make all cuts prescribed by dummy tuples preceding~$e_n, e_{n - 1}, \dots,e_{i+1}$; see Fig.~\ref{fig:vertex-selection-b}b.
  Note that these are $(m_G + 4) \cdot (n-1)$ cuts in this subspace, and in total~$(m_G + 4) \cdot (m_G-m_H)$ cuts in all of the~$m_H$ edge-selection subspaces.
  Additionally, note that each cut has one side on which there are only red or only blue examples, i.e., these cuts correspond to a decision tree whose internal nodes form a path.
  After performing these cuts for all edges of~$H$, all pairs of dummy examples associated with dummy tuples are separated from each other.
  The only primary example pairs that are not yet separated are those that correspond to the edges whose edge pairs have not been cut in the edge-selection subspaces. These examples are not separated from each other either, because they are not separated from the irrelevant examples in their respective subspaces.

  \looseness = -1
  Afterwards, for each vertex~$c$ of~$H$, we cut between the vertex pair of the vertex~$\phi(c) \in V(G)$ in the vertex-verification subspace~$S$ of~$c$.
  Observe that this separates both the dummy example~$D$ of~$S$ and one of the two pairs of examples that corresponds to the selected edges incident to~$\phi(c)$ (the other one is separated by the corresponding cut in the vertex-verification subspace of the other endpoint).
  Note that each of these $n_H$ cuts separates a set of red examples from the remaining examples, and they can be performed in an arbitrary order.
  In total, we have used~$(m_G + 4) \cdot (m_G-m_H)+n_H = s$ cuts, to separate the red from the blue examples, and in fact, the internal nodes of the decision tree form a path.

  $\Leftarrow$: Now assume that $(E,\lambda)$ admits a decision tree~$T$ with~$s$ cuts.
  First, observe that for each dummy tuple, the tree $T$ contains $m_{G} + 3$ distinct cuts by the way we have defined the projections of the dummy examples.
  In total, there are $(m_G + 3) \cdot (m_G - m_H)$ such cuts.
  Furthermore, since in each vertex-verification subspace there is a dummy pair, there are further $n_H$ distinct cuts in~$T$.
  Call all of these cuts \emph{required}.
  
  Say that an edge of $G$ is \emph{chosen} if its primary red and blue examples are not cut apart in an edge-selection subspace by~$T$.
  Among all solution decision trees, pick $T$ such that it has the minimum number of nodes that make cuts in vertex-verification subspaces.
  We claim that, for each edge $e$ of $H$, there is at most one chosen edge in $G$ whose endpoints have the same colors as~$e$.
  For a contradiction, assume that edges $f$ and $f'$ of $G$ are chosen and their endpoints' colors are the endpoints of~$e \in E(H)$.
  Since the primary examples of $f$ and $f'$ only differ in $S_e$ and in vertex-verification subspaces, these examples are cut apart in vertex-verification subspaces.
  Let $t, t'$ denote nodes in $T$ that cut the primary examples of $f$ and $f'$, respectively, via a cut in a vertex-verification subspace.
  Consider the lowest common ancestor~$a$ of $t$ and~$t'$.
  Observe that $E[a]$ contains both the primary examples of $f$ and $f'$.
  Thus, $E[a]$ also contains all examples that occur in the dummy tuples between $f$ and $f'$ in the edge-selection subspace $S_e$.

  Consider one such dummy tuple $D$ and let $b, b'$ be the children of $a$.
  Consider the case where $a$ equals $t$ or $t'$, that is, $a$ is a cut in a vertex-verification subspace corresponding to the endpoints of~$e$.
  Then, $a$ cuts each $R_i$ of $D$ from each $L_i$ of $D$.
  Thus for each $i \in \{1, 2, \ldots, m_G + 1\}$ it is impossible for $T$ to cut the pair $R_i$ and $L_i$ at the same time in the subtree of $T$ rooted at $a$ and thus, in addition to the required cuts, there are $m_G + 1$ further cuts in $T$.
  That is, the number of cuts in $T$ is at least $(m_G + 3) \cdot (m_G - m_H) + n_H + m_G + 1 > (m_G + 4) \cdot (m_G - m_H) + n_H = s$, a contradiction.
  Thus, $a$ is not equal to $t$ or $t'$.
  If $a$ does not cut at least one primary example of $f$ from at least one primary example of $f'$, then we may choose $t$ or $t'$ instead in a subtree rooted at $b$ or $b'$.
  Thus, assume that $a$ cuts at least one primary example of $f$ from at least one primary example of $f'$.
  Again, these primary examples only differ in $S_e$ and in vertex-verification subspaces.
  This implies that $a$ is either a cut in a vertex-verification subspace corresponding to an endpoint of $e$ or a cut in~$S_e$ somewhere between the primary examples of $f$ and $f'$.
  In the first case, we get a contradiction in the same way as for the case where $a$ is equal to $t$ or $t'$.
  So assume that $a$ is a cut in $S_e$ somewhere between the primary examples of $f$ and $f'$.
  Since the unlabeled point in $S_e$ is on exactly one side of this cut and all examples that are irrelevant for $S_e$ project to the unlabeled point, 
  it follows that either $E[b]$ or $E[b']$ contains only dummy examples or primary examples of edges with the same colors as~$e$; say $E[b]$ does.
  Assume that $t$ is in the subtree of $T$ rooted at~$b$.
  This is without loss of generality by renaming.
  Note that $E[t]$ cannot contain two pairs of primary examples corresponding to different edges:
  In this case, $E[t]$ contains also all examples of the dummy tuples between the two edges in $S_e$.
  Performing the cut at $t$ in a vertex-verification subspace would hence introduce $m_G + 1$ cuts in addition to the required ones (as before), a contradiction.
  Thus, $E[t]$ contains among primary examples only those of $f$.
  Hence, we may replace~$t$ by a cut in $S_e$, a contradiction, because $T$ has the minimum number of nodes that make cuts in vertex-verification subspaces.
  Thus, indeed, for each edge $e$ of $H$, there is at most one chosen edge in $G$ whose endpoints have the same colors as~$e$.

  By the calculation of required edges it now follows that also there is at least one chosen edge for each edge $e$ of $H$.
  For each edge $e$ in $H$ the primary examples of the chosen edge for $e$ need to be separated in both vertex-verification subspaces corresponding to the endpoints of~$e$.
  Furthermore, for each vertex-verification subspace there is one required cut in that subspace.
  Thus, two chosen edges whose endpoints have the same color are adjacent to the same vertex of that color.
  Thus, the chosen edges induce a subgraph isomorphism~$\phi$ from~$H$ to~$G$.
\end{proof}

\looseness = -1
Since the reduction takes polynomial-time, and $d = 2m_H+n_H \leq 4m_H$, \cref{thm:whardness} readily follows.

\section{Conclusion}
We have begun charting the boundary of tractability for learning small decision trees with respect to the number~$d$ of dimensions.
While exponents in the running time need to depend on~$d$, this dependency is captured by the number of leaves labeled with the first class, the class with the fewest leaves.
It would be interesting to analyse what other features of the input or output can capture the combinatorial explosion induced by dimensionality; this can be done by deconstructing our hardness result~\cite{komusiewicz_deconstructing_2011}.
Interesting parameters that are necessarily unbounded for the reduction to work include the number of examples that have the same feature values and the maximum number of alternations between labels when sorting the examples according to their feature values in a dimension.

\begin{ack}
Main ideas for the results of this paper were developed in the relaxed atmosphere of Dagstuhl Seminar 21062 on \emph{Parameterized Complexity in Graph Drawing}, organized by Robert Ganian, Fabrizio Montecchiani, Martin Nöllenburg, and Meirav Zehavi.
Manuel Sorge acknowledges funding by the Alexander von Humboldt Foundation. Jules Wulms acknowledges funding by the Vienna Science and Technology Fund (WWTF) under grant ICT19-035.
\end{ack}

\bibliographystyle{abbrvnat}
\bibliography{refs}

\begin{thebibliography}{24}
\providecommand{\natexlab}[1]{#1}
\providecommand{\url}[1]{\texttt{#1}}
\expandafter\ifx\csname urlstyle\endcsname\relax
  \providecommand{\doi}[1]{doi: #1}\else
  \providecommand{\doi}{doi: \begingroup \urlstyle{rm}\Url}\fi

\bibitem[Bessiere et~al.(2009)Bessiere, Hebrard, and
  O'Sullivan]{bessiere_minimising_2009}
C.~Bessiere, E.~Hebrard, and B.~O'Sullivan.
\newblock Minimising decision tree size as combinatorial optimisation.
\newblock In \emph{Proc. 15th International Conference on Principles and
  Practice of Constraint Programming (CP)}, pages 173--187, 2009.
\newblock \doi{10.1007/978-3-642-04244-7_16}.

\bibitem[Breiman et~al.(1984)Breiman, Friedman, Olshen, and
  Stone]{breiman_classification_1984}
L.~Breiman, J.~H. Friedman, R.~A. Olshen, and C.~J. Stone.
\newblock \emph{Classification and Regression Trees}.
\newblock Chapman \& Hall/CRC, 1984.

\bibitem[Cygan et~al.(2015)Cygan, Fomin, Kowalik, Lokshtanov, Marx, Pilipczuk,
  Pilipczuk, and Saurabh]{CyFoKoLoMaPiPiSa2015}
M.~Cygan, F.~V. Fomin, L.~Kowalik, D.~Lokshtanov, D.~Marx, M.~Pilipczuk,
  M.~Pilipczuk, and S.~Saurabh.
\newblock \emph{Parameterized Algorithms}.
\newblock Springer, 2015.

\bibitem[Downey and Fellows(2013)]{downey_fundamentals_2013}
R.~G. Downey and M.~R. Fellows.
\newblock \emph{Fundamentals of Parameterized Complexity}.
\newblock Texts in Computer Science. Springer, 2013.

\bibitem[Fayyad and Irani(1990)]{fayyad_what_1990}
U.~M. Fayyad and K.~B. Irani.
\newblock What should be minimized in a decision tree?
\newblock In \emph{Proc. 8th National Conference on Artificial Intelligence
  (AAAI)}, pages 749--754, 1990.
\newblock URL \url{https://www.aaai.org/Library/AAAI/1990/aaai90-112.php}.

\bibitem[Fellows et~al.(2009)Fellows, Hermelin, Rosamond, and
  Vialette]{fellows_parameterized_2009}
M.~R. Fellows, D.~Hermelin, F.~Rosamond, and S.~Vialette.
\newblock On the parameterized complexity of multiple-interval graph problems.
\newblock \emph{Theoretical Computer Science}, 410\penalty0 (1):\penalty0
  53--61, 2009.
\newblock \doi{10.1016/j.tcs.2008.09.065}.

\bibitem[Flum and Grohe(2006)]{flum_parameterized_2006}
J.~Flum and M.~Grohe.
\newblock \emph{Parameterized Complexity Theory}.
\newblock Springer, 2006.

\bibitem[Goodrich et~al.(1995)Goodrich, Mirelli, Orletsky, and Salowe]{GMOS95}
M.~T. Goodrich, V.~Mirelli, M.~Orletsky, and J.~Salowe.
\newblock Decision tree construction in fixed dimensions: Being global is hard
  but local greed is good.
\newblock Technical Report TR-95-1, Department of Computer Science, Johns
  Hopkins University, May 1995.

\bibitem[Gottlob et~al.(2002)Gottlob, Scarcello, and
  Sideri]{gottlob_fixedparameter_2002}
G.~Gottlob, F.~Scarcello, and M.~Sideri.
\newblock Fixed-parameter complexity in {AI} and nonmonotonic reasoning.
\newblock \emph{Artif. Intell.}, 138\penalty0 (1-2):\penalty0 55--86, 2002.
\newblock \doi{10.1016/S0004-3702(02)00182-0}.

\bibitem[Hu et~al.(2019)Hu, Rudin, and Seltzer]{DBLP:conf/nips/HuRS19}
X.~Hu, C.~Rudin, and M.~I. Seltzer.
\newblock Optimal sparse decision trees.
\newblock In \emph{Proc. 33th Annual Conference on Neural Information
  Processing Systems (NeurIPS)}, pages 7265--7273, 2019.
\newblock URL
  \url{https://proceedings.neurips.cc/paper/2019/hash/ac52c626afc10d4075708ac4c778ddfc-Abstract.html}.

\bibitem[Hyafil and Rivest(1976)]{hyafil_constructing_1976}
L.~Hyafil and R.~L. Rivest.
\newblock Constructing optimal binary decision trees is np-complete.
\newblock \emph{Information Processing Letters}, 5\penalty0 (1):\penalty0
  15--17, 1976.
\newblock \doi{10.1016/0020-0190(76)90095-8}.

\bibitem[Impagliazzo and Paturi(2001)]{impagliazzo_complexity_2001}
R.~Impagliazzo and R.~Paturi.
\newblock On the complexity of $k$-{SAT}.
\newblock \emph{Journal of Computer and System Sciences}, 62\penalty0
  (2):\penalty0 367--375, 2001.
\newblock \doi{10.1006/jcss.2000.1727}.

\bibitem[Impagliazzo et~al.(2001)Impagliazzo, Paturi, and
  Zane]{impagliazzo_which_2001}
R.~Impagliazzo, R.~Paturi, and F.~Zane.
\newblock Which problems have strongly exponential complexity?
\newblock \emph{Journal of Computer and System Sciences}, 63\penalty0
  (4):\penalty0 512--530, 2001.
\newblock \doi{10.1006/jcss.2001.1774}.

\bibitem[Komusiewicz et~al.(2011)Komusiewicz, Niedermeier, and
  Uhlmann]{komusiewicz_deconstructing_2011}
C.~Komusiewicz, R.~Niedermeier, and J.~Uhlmann.
\newblock Deconstructing intractability\textemdash a multivariate complexity
  analysis of interval constrained coloring.
\newblock \emph{Journal of Discrete Algorithms}, 9\penalty0 (1):\penalty0
  137--151, 2011.
\newblock \doi{10.1016/j.jda.2010.07.003}.

\bibitem[Marx(2010)]{marx_can_2010}
D.~Marx.
\newblock Can you beat treewidth?
\newblock \emph{Theory of Computing}, 6\penalty0 (1):\penalty0 85--112, 2010.
\newblock \doi{10.4086/toc.2010.v006a005}.

\bibitem[Molnar(2020)]{molnar_interpretable_2020}
C.~Molnar.
\newblock \emph{Interpretable Machine Learning}.
\newblock Independently published, 2020.
\newblock URL \url{https://christophm.github.io/interpretable-ml-book}.

\bibitem[Moshkovitz et~al.(2020)Moshkovitz, Dasgupta, Rashtchian, and
  Frost]{moshkovitz_explainable_2020}
M.~Moshkovitz, S.~Dasgupta, C.~Rashtchian, and N.~Frost.
\newblock Explainable k-means and k-medians clustering.
\newblock In \emph{Proc. 37th International Conference on Machine Learning
  (ICML)}, pages 7055--7065, 2020.
\newblock URL \url{http://proceedings.mlr.press/v119/moshkovitz20a.html}.

\bibitem[Murthy(1998)]{murthy_automatic_1998}
S.~K. Murthy.
\newblock Automatic construction of decision trees from data: A
  multi-disciplinary survey.
\newblock \emph{Data Mining and Knowledge Discovery}, 2\penalty0 (4):\penalty0
  345--389, 1998.
\newblock \doi{10.1023/A:1009744630224}.

\bibitem[Narodytska et~al.(2018)Narodytska, Ignatiev, Pereira, and
  {Marques-Silva}]{narodytska_learning_2018}
N.~Narodytska, A.~Ignatiev, F.~Pereira, and J.~{Marques-Silva}.
\newblock Learning optimal decision trees with sat.
\newblock In \emph{Proc. 27th International Joint Conference on Artificial
  Intelligence (IJCAI)}, pages 1362--1368, 2018.
\newblock \doi{10.24963/ijcai.2018/189}.

\bibitem[Niedermeier(2006)]{Nie06}
R.~Niedermeier.
\newblock \emph{Invitation to Fixed-Parameter Algorithms}.
\newblock Oxford, 2006.

\bibitem[Ordyniak and Szeider(2021)]{DBLP:conf/aaai/OrdyniakS21}
S.~Ordyniak and S.~Szeider.
\newblock Parameterized complexity of small decision tree learning.
\newblock In \emph{Proc. 35th {AAAI} Converence on Artificial Intelligence
  (AAAI)}, pages 6454--6462, 2021.
\newblock URL \url{https://ojs.aaai.org/index.php/AAAI/article/view/16800}.

\bibitem[Schidler and Szeider(2021)]{schidler_satbased_2021}
A.~Schidler and S.~Szeider.
\newblock {SAT}-based decision tree learning for large data sets.
\newblock \emph{Proc. 35th {AAAI} Conference on Artificial Intelligence
  (AAAI)}, 35\penalty0 (5):\penalty0 3904--3912, 2021.
\newblock URL \url{https://ojs.aaai.org/index.php/AAAI/article/view/16509}.

\bibitem[Steinberg(2009)]{steinberg_cart_2009}
D.~Steinberg.
\newblock {CART}: Classification and regression trees.
\newblock In X.~Wu and V.~Kumar, editors, \emph{The Top Ten Algorithms in Data
  Mining}, pages 179--201. Chapman \& Hall/CRC, 2009.
\newblock \doi{10.1201/9781420089653}.

\bibitem[Wu et~al.(2008)Wu, Kumar, Ross~Quinlan, Ghosh, Yang, Motoda,
  McLachlan, Ng, Liu, Yu, Zhou, Steinbach, Hand, and Steinberg]{wu_top_2008}
X.~Wu, V.~Kumar, J.~Ross~Quinlan, J.~Ghosh, Q.~Yang, H.~Motoda, G.~J.
  McLachlan, A.~Ng, B.~Liu, P.~S. Yu, Z.-H. Zhou, M.~Steinbach, D.~J. Hand, and
  D.~Steinberg.
\newblock Top 10 algorithms in data mining.
\newblock \emph{Knowledge and Information Systems}, 14\penalty0 (1):\penalty0
  1--37, 2008.
\newblock \doi{10.1007/s10115-007-0114-2}.

\end{thebibliography}

\end{document}